\documentclass[authoryear]{amsart}
\usepackage{float}
\usepackage{amsmath, amssymb}
\usepackage{graphicx}
\usepackage{latexsym}
\usepackage{amsfonts,amsthm}
\usepackage{color}
\usepackage{dsfont}
\usepackage{pgf}
\usepackage{blkarray}
\usepackage{tikz}
\usepackage{standalone}
\usepackage{natbib}

\usepackage[utf8]{inputenc}
\graphicspath{{figs/}}

\usepackage{ulem}

\usetikzlibrary{arrows}
\usepackage{ifthen}

\usepackage{comment}

\newtheorem{proposition}{Proposition}
\newtheorem{conjecture}{Conjecture}
\newtheorem{corollary}{Corollary}
\newtheorem{remark}{Remark}

\newcommand\1{\mathds{1}}

\newcommand{\hP}{\hat{P}}

\def\E{{\mathbb E}}

\def\S{{\mathcal{S}}}

\newcommand{\bt}{\tilde{\beta}}

\newcommand{\bn}[1]{\beta_{#1}(\nu)}

\newcommand{\py}{\frac{\partial}{\partial y}}
\newcommand{\px}{\frac{\partial}{\partial x}}

\newcommand{\fs}{\textit{fluid stable}}
\newcommand{\fu}{\textit{fluid unstable}}
\newcommand{\no}{\textit{normal}}

\newcommand{\sm}{\bar{s}}

\begin{document}

\title{Markovian queues with Poisson control}
\thanks{A short letter mainly devoted to Conjecture 1 of this paper appeared in \citet{NPR22}.}

\author{R.~Núñez-Queija}
\address[R. Núñez-Queija]{University of Amsterdam, Amsterdam, The Netherlands}
\email{nunezqueija@uva.nl}

\author{B.J.~Prabhu}
\address[B.J.~Prabhu]{LAAS-CNRS, Universit\'e de Toulouse, CNRS, INSA, Toulouse, France}
\email{balakrishna.prabhu@laas.fr}

\author{J.A.C.~Resing}
\address[J.A.C.~Resing]{Eindhoven University of Technology, Eindhoven, The Netherlands}
\thanks{J.A.C.~Resing is the corresponding author}
\email{j.a.c.resing@tue.nl}

\begin{abstract}
We investigate Markovian queues that are examined by a controller at random times determined by a Poisson process. Upon examination, the controller sets the service speed to be equal to the minimum of the current number of customers in the queue and a certain maximum service speed; this service speed prevails until the next examination time. We study the resulting two-dimensional Markov process of queue length and server speed, in particular two regimes with time scale separation, specifically for infinitely frequent and infinitely long examination times. In the intermediate regime the analysis proves to be extremely challenging. To gain further insight into the model dynamics we then analyse two variants of the model in which the controller is just an observer and does not change the speed of the server.
\end{abstract}

\maketitle


\section{Introduction}
Consider a dynamically controlled queue in which the speed of the server can be varied in order to attain a given objective. The optimal policy for such problems usually requires adjustments of the speed based on the current queue-length. For example, it has been shown that the right trade-off between energy consumption and mean sojourn time is achieved by setting the speed to be proportional to the number of customers to the power of a constant \citep{Bansal:2007, Andrew:2010}. These adjustments to the speed have to be made at arbitrary times which may depend upon the dynamics of the queue-length (for example, at arrival instants or departure instants). However, in practice  communicating real-time state information of the queue may not always be feasible or desirable due to communication overheads. Nevertheless, it may be possible to control the speed at, say, Poisson instants which are independent of the state of the queue. These control instants can be seen as moments at which a controller, oblivious of the state, decides to measure the state and take appropriate action.

In this paper, we analyse a queueing model in which the server speed is modified at Poisson instants of rate $\nu$ that are independent of the queue length evolution. Inspired from the policy in \citet{Bansal:2007, Andrew:2010}, it will be assumed that the speed is set to $\min(q, \sm)$ when $q$ is the observed queue length and $\sm$ the maximum server speed. 
Our results can be used to compute the sub-optimality induced in the performance metrics due to control instants being different from the ones prescribed by the optimal policy. We shall call this model the \textit{Poisson controller} model.

Our main focus is on the model with $\sm = \infty$. The analysis of that model is not trivial as we show in the paper and for that reason we use the finite speed model to gain insight into the reasons for that complexity.

\subsection{Contributions}
The Poisson controller model can be described by a two-dimensional Markov process with state given by ($Q(t), S(t)$) where $Q(t)$ denotes the number of customers in the system at time $t$ and $S(t)$ denotes the speed of the server at $t$.

For the analysis, we shall separate the two cases $\sm < \infty$ and $\sm = \infty$ since their treatment is based on different techniques. Moreover, as will be seen later, they give rise to different asymptotic results which are easier to present separately.

 For the infinite maximum speed case (Section \ref{sec:infspeed}), that is $\sm = \infty$, we determine the functional equation whose solution leads to the joint generating function of the steady-state process. The steady-state distribution of the queue-length conditional on the speed is shown to be equivalent to the transient distribution of a queue whose state is reset to the given speed at the control instants. The latter queue was analysed in \citet{cohen1982} in which he gave the generating function for the steady-state queue-length. From this  conditional generating function as well as the observation that the marginal distribution of $Q$ is the same as that of $S$, we obtain a system of linear equations to compute this marginal distribution. Finally, we investigate the limiting behaviour of the steady-state distribution when the rate of the control process, $\nu$, goes to $0$ and to $\infty$. For $\nu\to 0$, it is possible that the queue is unstable for some of the speeds. The queue-length can thus live on two different scales: (i) the fluid scale when the sampled speed is less than the arrival rate; (ii) the normal scale when the sampled speed is greater than the arrival rate.

For finite maximum speed $\sm$ (Section \ref{sec:finspeed}), we resort to the matrix geometric method \citep{latouche1999} to analyze the system. Using a probabilistic interpretation, we obtain an explicit expression for the $R$ matrix. As an illustration, we derive the steady-state probabilities for $\sm=1$. For the general case, we show how to obtain the joint generating function when $\nu\to 0$ and $\nu\to\infty$.

We end the paper with two variants of the model, the $M/M/1$ model and the $M/M/\infty$ model, in which the controller is just an observer and does not change the speed of the server. 
The reason that we look at these variants is that in these models, like in the model with Poisson controller, the two-dimensional Markov process also has jumps to diagonal states. However, contrary to the model with Poisson controller, for these variants an explicit expression of the joint generating function of the steady-state process, keeping track of the current queue length and the last observed queue length in this case, can be obtained by solving the corresponding functional equation.

\subsection{Related work}
Our model is related to queueing systems or Markov chains in random environments \citep{Lovas2021,Liu2021}. In the cited models, the arrival rate or the service rate of the queue depends upon the state of the environment which is a random process. As opposed to this, in our model, the environment (which  is the speed observed by the controller) itself depends upon the queue length since it is set to the value measured at the control instants. The two variables -- the queue-length and the state -- thus influence the dynamics of each other in our model whereas in the classical random environment model it is only the environment that influences the dynamics of queue-length. We also mention \citet{cheung2010} in which a random environment model with both unstable and stable speeds was analysed. They obtained the conditional generating function for the queue-length for the process on both the fluid as well as the normal scale. We obtain this type of results for the Poisson controller model. 

Another related model is the model with workload dependent arrival and service rates in \citet{Bekker04}. In that work, the rates change instantaneously with the state which is not the case for us.
The finite speed Poisson controller model with only two speeds (in Section \ref{sec:finspeed}) is a special case of the queue with service speed adaptations  analysed in \citet{B2008}.

\section{Poisson controller with infinite maximum speed}
\label{sec:infspeed}
Consider a single-server queue to which arrivals occur according to a Poisson process of rate $\lambda$. Each arrival brings with it an exponentially distributed service requirement with mean $1/\mu$. The speed of the server can be dynamically assigned values in the set $\{0,1,2,\hdots,\}$. Adjustments to the speed can be made only at control instants which are assumed to occur
according to a Poisson process of rate $\nu$ independently of the arrivals and the departures. At a control instant, the speed of the server is set equal to the number of customers observed at that instant. Between any two consecutive control instants, the speed of the server remains constant at the value chosen at the earlier control instant.

Let $Q(t)$ denote the number of customers in the system at time $t$ and $S(t) = Q(\tau_t)$, with $\tau_t$ the last control instant at or before time $t$. Thus, the speed at time $t$ is maintained at $Q(\tau_t)$ until the next control instant. We do not specify the initial state $(Q(0), S(0))$ since we are interested in the steady-state behaviour of the system which is independent of the initial state. Furthermore, as our focus is on the number of customers in the system and service times are exponential, the service discipline can be any discipline in which the service order is independent of the actual service times of the customers (e.g., FCFS, LCFS or ROS).

The  process $(Q(t),S(t))_{t\geq 0}$ is a Markov process with transition rates
\begin{equation}
(Q(t), S(t)) \to \left\{ \begin{array}{lcl}
			(Q(t) + 1, S(t)) & \mbox{with rate} & \lambda; \\
			(Q(t) - 1, S(t)) & \mbox{with rate} & \mu S(t); \\
			(Q(t), Q(t)) & \mbox{with rate} & \nu. \\
			\end{array}\right.
\end{equation}
Figure \ref{fig:transition_diag} shows the rate diagram of this Poisson controller model.
 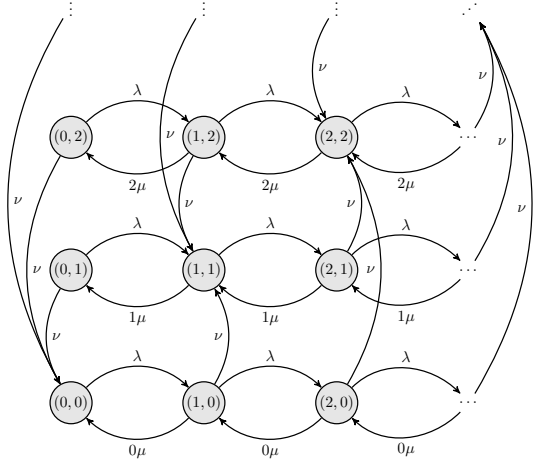
\begin{figure}[!htb]
  \centering
  \tikzstyle{input}=[font=\small\sffamily\bfseries]
\tikzstyle{rect}=[rectangle, draw=black, font=\small\sffamily\bfseries, inner sep=9pt]
\begin{tikzpicture}[->,>=stealth',auto,node distance=3cm,
  thick,main node/.style={fill=gray!20,draw, circle,inner sep=1pt}, ph node/.style={draw=none}]
  
  \foreach \x in {0,...,3}
    \foreach \y in {0,...,3} {
        \ifthenelse{\x=3 \OR \y=3}{
        		\ifthenelse{\x=3 \AND \y=3}{\node[draw=none] (\x\y) at (3*\x,3*\y) {\reflectbox{$\ddots$}};}{
		 	\ifthenelse{\x = 3}{\node[draw=none] (\x\y) at (3*\x,3*\y){$\cdots$};}{\node[draw=none] (\x\y) at (3*\x,3*\y){$\vdots$};};};
			}{\node [main node]  (\x\y) at (3*\x,3*\y) {$(\x, \y)$};};
       };
      
%
  
\foreach \x [count=\xi] in {0,...,2}
    \foreach \y [count=\yi] in {0,...,2}  
      {\draw [->] (\x\y) to [out=45,in=135] node [pos=0.5,above]  {$\lambda$} (\xi\y);
           \draw [->] (\xi\y) to [out=225,in=315] node [pos=0.5, below] {$\y \mu$}(\x\y);
       };
\foreach \x in {0,...,3}
        \foreach \y in {0,...,3}{ 
		 \ifthenelse{\x=\y}{}{\ifthenelse{\x > \y}{\draw [->] (\x\y) to [bend right] node [pos=0.5, left] {$\nu$}(\x\x);}{\draw [->] (\x\y) to [bend right] node [pos=0.5, right] {$\nu$}(\x\x);};};
	}

\end{tikzpicture}
    \caption{Rate diagram of the model with Poisson controller with infinite maximum speed.}
    \label{fig:transition_diag}
  \end{figure}

The process $(Q(t), S(t))_{t\geq 0}$ is ergodic for all possible combinations $\lambda>0$, $\mu>0$ and $\nu>0$ (see also Remark \ref{rem:ergodic} in Section 3) and we are interested in the steady-state behaviour of this two-dimensional Markov process. Denote with $\pi_{i,j} = \lim_{t\to\infty} P(Q(t)=i, S(t)=j)$ the steady-state probabilities and let

\begin{equation}
	P(x,y) = \sum_{i\geq 0,j\geq 0} \pi_{i,j} x^i y^j
\end{equation}
be the corresponding joint probability generating function. By $(Q,S)$ we denote a pair of random variables with this joint probability generating function.

\begin{proposition}
$P(x,y)$ is the solution of the functional equation
\begin{equation}
	(\nu + \lambda(1-x))P(x,y) + \mu y \left(1 - \frac{1}{x}\right) \py\left[P(x,y) - P(0,y)\right] = \nu P(xy,1).
\label{eqn:gen_fun}
\end{equation}
\label{prop:gen_fun}
\end{proposition}
\begin{proof}

From the transition in Fig.~\ref{fig:transition_diag}, the balance equations for the stationary probabilities can be seen to be, for $i \geq 1, j\geq 0$:
\begin{align}
	(\lambda + j\mu + \nu)\pi_{i,j} &= \lambda\pi_{i-1,j} + j\mu \pi_{i+1,j} + \1_{\{i=j\}}\nu\sum_{k}\pi_{j,k} 
 \label{eqn:balpi1}
\end{align}
and, for $i=0, j\geq 0$:
\begin{align}
	(\lambda + \nu)\pi_{i,j} &= j\mu \pi_{i+1,j} + \1_{\{i=j\}}\nu\sum_{k}\pi_{j,k}.
 \label{eqn:balpi2}  
\end{align}

Multiplying both sides of the rate balance equations by $x^iy^j$ and summing
over all possible $i$ and $j$ immediately leads to the equation
\begin{eqnarray*}
&&	(\nu + \lambda)P(x,y) + \mu y \py\left[P(x,y) - P(0,y)\right] \\
&&= \lambda x P(x,y) + \frac{\mu y}{x} \py\left[P(x,y) - P(0,y)\right] + \nu P(xy,1),
\end{eqnarray*}
which can be alternatively written as (\ref{eqn:gen_fun}).
\end{proof}
\begin{corollary}
The marginal distribution of the speed is the same as the marginal distribution of the number of customers in the system. That is,
\[P(1,y) = P(y,1).\]
\label{cor:marg}
\end{corollary}
\begin{proof}
Substituting $x=1$ in Prop. \ref{prop:gen_fun} leads to the above corollary.
\end{proof}

An explanation for Cor.~\ref{cor:marg} is that the controller
sees the marginal distribution of $Q$ due to the PASTA property\footnote{
We use that "PASTA" holds for any sequence of Poisson-generated events that may or may not change the system state (see \cite{W1982}). In the literature this property has been coined PASTA because it is mostly used for queueing systems at arrival instants.}.

Since $S$ is set to $Q$ at the control instants, the marginal distribution
of $S$ just after these instants is the same as the marginal distribution of $Q$.
Because the time until the next control instant is independent
of $S$ at a control instant (this time is exponentially distributed
with parameter $\nu$), the marginal distribution of $S$ at an arbitrary
instant is the same as the marginal distribution of $S$ just after a
control instant and hence also the same as the marginal distribution
of $Q$ at an arbitrary instant.

Let $\gamma_i = \sum_j\pi_{i,j}$ be the marginal distribution of the stationary queue-length.
A second consequence of Prop.~\ref{prop:gen_fun} is the following local balance when there are $i$ customers in the queue.
\begin{corollary}
\begin{equation}
	\lambda\gamma_i = \sum_j j\mu\pi_{i+1,j},\, i\geq0.
\end{equation}
\label{cor:avgrate}
\end{corollary}
\begin{proof}
Substituting $y=1$ in the Prop. \ref{prop:gen_fun}  gives
\begin{equation}
	x\lambda P(x,1) = \mu \py\left[P(x,y) - P(0,y)\right]\Big\rvert_{y=1}.
\label{eq:subst_y=1}
\end{equation}
Comparing the coefficient of $x^i$ on either side we get the claimed result.
\end{proof}
An alternative proof of this result can be obtained using the following up- and
downcrossings argument. The rate at which the number of customers in the
system goes from $i$ to $i+1$ equals the rate at which the number of
customers goes from $i+1$ to $i$.

Finally, we can also conclude that the expected queue-length is larger than in an $M/M/\infty$ queue.
\begin{corollary}
\[
\E[Q] > \rho, \mbox{ with $\rho = \lambda/\mu$}.
\]
\end{corollary}
\begin{proof}
Substituting $x=1$ in \eqref{eq:subst_y=1} yields the equation
\[
\lambda =  \mu \left( E(S) - E(S\cdot 1[Q=0])\right),
\]
from which we conclude that
\[
E(S \cdot 1[Q>0]) = \rho,
\]
and, hence,
\[
E[Q] = E[S] = E(S \cdot 1[Q>0]) + E(S\cdot 1[Q=0]) > \rho.
\]
\end{proof}

We were unable to obtain an explicit solution to the functional equation \eqref{eqn:gen_fun}. Next, we provide an alternative way to compute the generating function. This method is not explicit but is amenable to numerical computations. For this, we first compute the distribution of the queue-length conditioned on the speed.
\subsection{The conditional distribution}
\label{ssec:cond_dist}
Define $\sigma_j= \sum_i\pi_{i,j}$ to be the marginal distribution of the stationary server speed. Let 
\[
p_j(i) = \frac{\pi_{i,j}}{\sigma_j}
\]
be the stationary conditional probability of having $i$ customers in the system when the service rate is $j\mu$ and let
$f_j(z)$ be the generating function of this stationary conditional distribution. Furthermore, define 	
\begin{equation}
\beta_j(\nu) = \frac{\lambda+j\mu+\nu - \sqrt{(\lambda+j\mu + \nu)^2 - 4\lambda(j\mu)}}{2\lambda}
\end{equation}

\begin{proposition}
For the generating function $f_j(z)$ of the stationary conditional distribution we have
\begin{equation}
f_0(z) =  \frac{\nu}{\nu+\lambda(1-z)}
\label{eqn:mar_0}
\end{equation}
and, for $j>0$,
\begin{equation}
	f_j(z) = \frac{\nu \bt_j(\nu)\sum_{k=0}^{\infty} c_{k,j} z^k}{\lambda(1- \bt_j(\nu)z)},
	\label{eqn:mar_q}
\end{equation}
where
\begin{equation}
	\bt_j(\nu) = \frac{\lambda\beta_j(\nu)}{j\mu},
\end{equation}
and
\begin{equation}
	c_{k,j} = \begin{cases}
			(1-\beta_j(\nu))^{-1}\beta_j(\nu)^{j} & k = 0; \\
			\beta_j(\nu)^{j-k} & k = 1, \ldots, j; \\
			0 & k > j.
	\end{cases}
\label{eqn:c_kj}
\end{equation}
\label{prop:cond_distr}
\end{proposition}
\begin{proof}
For the conditional process the following balance equations can be obtained by dividing \eqref{eqn:balpi1} and \eqref{eqn:balpi2} by $\sigma_j$ and by noting that the $\sum_k \pi_{j,k} = \sigma_j$ (from Cor. \ref{cor:marg} the marginal distribution of the queue-length is the same as the marginal distribution of the speed):
\begin{align}
(\lambda + \nu)p_j(0) &= j\mu p_j(1), \nonumber \\
	(\lambda + \nu + j\mu) p_j(i) &= \lambda p_j(i-1) + j\mu p_j(i+1),  \quad i > 0, \quad  i\neq j; \nonumber \\
	(\lambda + \nu + j\mu) p_j(j) &= \lambda p_j(j-1)+ j\mu p_j(j+1) + \nu. \nonumber
\end{align}
Hence the generating function satisfies
\[
(\lambda + \nu + j \mu) f_j(z) - j \mu p_j(0) = \lambda z f_j(z) + \nu z^j + \tfrac{j\mu}{z}\left[f_j(z)-p_j(0)\right],
\]
which leads to the equation
\begin{equation}
\left[\lambda z^2 - (\lambda + \nu + j \mu) z + j\mu\right] f_j(z) = j\mu (1-z) p_j(0) - \nu z^{j+1}.
\label{eq:f_j(z)}
\end{equation}

For $j=0$, (\ref{eq:f_j(z)}) leads to
\[
f_0(z) = \frac{-\nu z}{\lambda z^2 - (\lambda + \nu) z} = \frac{\nu}{\lambda + \nu - \lambda z}  = \frac{\nu}{\nu+\lambda(1-z)}.
\]
For $j>0$, let $z_1 < 1$ and $z_2>1$ be the zeros of the polynomial in the lefthandside of \eqref{eq:f_j(z)}.
For $z=z_1$ also the righthandside of (\ref{eq:f_j(z)}) should be zero, hence
\[
p_j(0) = \frac{\nu z_1^{j+1}}{j \mu (1-z_1)}.
\]
So we conclude that
\begin{eqnarray*}
f_j(z) &=& \frac{\nu (1-z) z_1^{j+1} - \nu (1-z_1) z^{j+1}}{\lambda(z-z_1)(z-z_2)(1-z_1)} \\
&=& \frac{(\nu/z_2)}{\lambda(1-z/z_2)} \cdot \frac{(1-z_1) z^{j+1} - (1-z) z_1^{j+1}}{(z-z_1)(1-z_1)} \\
&=& \frac{(\nu/z_2)}{\lambda(1-z/z_2)} \cdot \left(\sum_{k=1}^j z_1^{j-k} z^k + \frac{z_1^j}{1-z_1}\right).
\end{eqnarray*}
Now, clearly, $z_1=\beta_j(\nu)$ and $\lambda z_1 z_2 = j \mu$. Hence, $z_2 = (j\mu)/(\lambda z_1) = 1 / \bt_j(\nu)$ and the result follows.
\end{proof}

\begin{remark}
The stationary queue-length process conditional on $S=j$ can readily be seen to have the stationary distribution of a process $Q_j(\cdot)$, obtained from the ordinary queue length process after eliminating time intervals for which the speed $\neq j$. Then, $Q_j(\cdot)$ behaves as the number of customers in an $M/M/1$ queue with arrival rate $\lambda$, service rate $j\mu$, and restarting in state $j$ after $\exp(\nu)$ distributed time periods.
It is a one-dimensional Markov chain with rate diagram shown in Fig.~\ref{fig:transition_diag_cond}.
 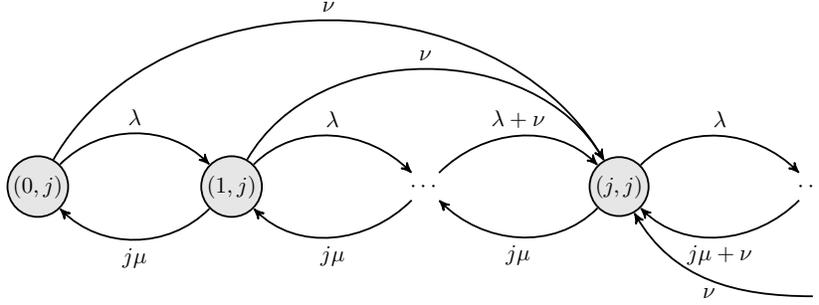
\begin{figure}[!htb]
  \centering
  \tikzstyle{input}=[font=\small\sffamily\bfseries]
\tikzstyle{rect}=[rectangle, draw=black, font=\small\sffamily\bfseries, inner sep=9pt]
\begin{tikzpicture}[->,>=stealth',auto,node distance=3cm,
  thick,main node/.style={fill=gray!20,draw, circle,inner sep=1pt}, ph node/.style={draw=none}]
  
  \foreach \x in {0,...,1} {
  	\node [main node]  (\x) at (3*\x,0) {$(\x, j)$};
  }
  \node[draw=none] (2) [right of=1] {$\cdots$};
  \node [main node]  (3) [right of=2] {$(j, j)$};
  \node[draw=none] (4) [right of=3] {$\cdots$};
  \coordinate (d1) at (12,-1.7);
	\foreach \x [count=\xi] in {0,...,3}
	      {
	      \ifthenelse{\x = 2}
	      		{\draw [->] (\x) to [out=45,in=135] node [pos=0.5,above]  {$\lambda+\nu$} (\xi);}
	      		{\draw [->] (\x) to [out=45,in=135] node [pos=0.5,above]  {$\lambda$} (\xi);};
			
		\ifthenelse{\xi=4}
			 {\draw [->] (\xi) to [out=225,in=315] node [pos=0.5, below] {$j \mu+\nu$}(\x);}
	           	{\draw [->] (\xi) to [out=225,in=315] node [pos=0.5, below] {$j \mu$}(\x);};
	           
	           \ifthenelse{\x < 2}{\draw [->] (\x) to [out=60, in =120] node [pos=0.5, above] {$\nu$}(3);}{};
	       };
	       
	       \draw [->] (d1) to [out=180, in =300] node [pos=0.5, below] {$\nu$}(3);
  

\end{tikzpicture}
    \caption{Rate diagram of the conditional process.}
    \label{fig:transition_diag_cond}
  \end{figure}

Further, the stationary distribution of $Q_j$ is the same as the transient distribution after an exponentially distributed time period with parameter $\nu$ in an $M/M/1$ queue with arrival rate $\lambda$, service rate $j \mu$ and starting at time $0$ in state $j$. Hence, equation (\ref{eqn:mar_q})
can be alternatively derived using results on the transient distribution in the $M/M/1$ queue (see Section I.4.4 and in particular formula (4.27) in \citet{cohen1982}).
\end{remark}
\begin{remark} 
From Proposition \ref{prop:cond_distr} it follows that, for $j>0$, we have that
\[
(Q \, | \, S=j) \overset{d}{=} \left[j-A_j\right]_+ + B_j,
\]
where $A_j$ and $B_j$ are independent random variables with $A_j$ geometrically distributed with parameter $\beta_j(\nu)$ and $B_j$ geometrically distributed with parameter $\tilde{\beta}_j(\nu)$.
\end{remark}

From the conditional generating function we can obtain the following expression for the stationary conditional distribution.
\begin{corollary}
For $j>0$,
\begin{equation}
	f_j(z) = \frac{\nu \bt_j(\nu)}{\lambda}\sum_{l\geq 0} \left(\sum_{k=0}^l c_{k,j}\bt_j(\nu)^{l-k}\right)z^l,
\label{eqn:mar_q_2}
\end{equation}
and hence
\begin{equation}
        p_j(l) = \frac{\nu}{\lambda} \sum_{k=0}^l c_{k,j}\bt_j(\nu)^{l-k+1}.
\label{eqn:mar_q_3}
\end{equation}
\end{corollary}
The mean conditional queue-length can also be expressed explicitly.
\begin{corollary}
\begin{align}
	\E[Q_j] = f^\prime_j(1) &= \frac{\nu \bt_j(\nu)}{\lambda} \left(\frac{\bt_j(\nu)\sum_k c_{k,j}}{(1-\bt_j(\nu))^2} +
\frac{\sum_k kc_{k,j}}{1-\bt_j(\nu)} \right)
\label{eqn:eqj}
\end{align}	
with
\begin{equation}
		\sum_k c_{k,j} = (1-\bn{j})^{-1},
	\label{eqn:sckj}
\end{equation}
and 
\begin{equation}
\sum_k k c_{k,j} = \frac{j  - (j+1) \beta_j(\nu) +\beta_j(\nu)^{j+1}}{(1-\beta_j(\nu))^2}
\label{eqn:skckj}
\end{equation}
\end{corollary}

\begin{proof}
The expression in \eqref{eqn:eqj} follows directly. For \eqref{eqn:sckj}, from \eqref{eqn:c_kj} for $c_{k,j}$ we get
\begin{equation}
		\sum_k c_{k,j} = (1-\bn{j})^{-1}.
\end{equation}
Alternatively, since $f_j(1) = 1$, we have
\begin{equation}
	1 = \frac{\nu \bt_j(\nu)\sum_k c_{k,j}}{\lambda(1- \bt_j(\nu))}	
\end{equation}
and hence
\begin{equation}
\sum_k c_{k,j} = \frac{\lambda(1- \bt_j(\nu))}{\nu \bt_j(\nu)} = \frac{j\mu - \lambda \bn{j}}{\nu \bn{j}}.
\end{equation}
These two expressions are the same because $\lambda \bn{j}^2 - (\lambda + \nu + j \mu) \bn{j} + j\mu =0$.
For \eqref{eqn:skckj}, we have
\begin{equation}
\sum_k k c_{k,j} = \sum_{k=1}^{j} k \bn{j}^{j-k} = \frac{j  - (j+1) \beta_j(\nu) +\beta_j(\nu)^{j+1}}{(1-\beta_j(\nu))^2}
\end{equation}

\end{proof}

\begin{proposition}
\begin{equation}
	\sigma_l = \sum_j\sigma_j\frac{\nu \bt_j(\nu)}{\lambda}\left(\sum_{k=0}^l c_{k,j}\bt_j(\nu)^{l-k}\right), \, \forall l \geq 0.
\label{eqn:sigma}
\end{equation}
\label{prop:sigma}
\end{proposition}
\begin{proof}
Observe that we can rewrite $P(x,y)$ in terms of $f_j(x)$ in the following way.
\begin{equation}
P(x,y) = \sum_j \sigma_j f_j(x) y^j.
\label{eqn:altP}
\end{equation}
Now use $P(1,y)=P(y,1)$ to obtain
\begin{equation}
		\sum_j \sigma_j f_j(1) y^j =  \sum_j \sigma_j f_j(y).
\label{eqn:sig_S}
\end{equation}
Substituting \eqref{eqn:mar_q_2}) in \eqref{eqn:sig_S} and comparing the coefficient of $y^l$ on either side, we get the desired system of linear equations for $\sigma_j$.
\end{proof}
\begin{remark}
Equation \eqref{eqn:sigma} can be interpreted as the balance equation of the embedded Markov chain of the server speed at control instants (see also \eqref{eqn:mar_q_3}).
\end{remark}

Since \eqref{eqn:sigma} is an infinite set of linear equations, it is not straightforward to solve them. Therefore, we look at the asymptotics of $P(x,y)$ when the rate of control $\nu$ goes to either $\infty$ or to $0$.

\subsection{Asymptotics for $\nu \to \infty$.}
When the rate of control $\nu \to \infty$, intuitively one expects the speed to be the same as the queue length since measurements are being made at a much faster rate compared to rates of variations in the queue-length. The two-dimensional process will live mainly on the diagonal states. The following result formalizes this intuition.

\begin{proposition}
If $\nu \to \infty$, then $P(x,y) \to e^{\rho(xy-1)}$.
\end{proposition}
\begin{proof}
First of all, remark that, when $\nu \to \infty$, functional equation (\ref{prop:gen_fun}) reduces to $P(x,y)=P(xy,1)$ from which we conclude that $\pi_{i,j} \to 0$ for $i \neq j$.
Next, assume that $\pi_{i,j}$ has the following analytic expansion
\begin{equation}
	\pi_{i,j} = \sum_m \pi^{(m)}_{i,j}\nu^{-m}.
\end{equation}
Then we have that $\pi^{(0)}_{i,j}=0$ for $i \neq j$ and furthermore $\pi^{(1)}_{i,j}=0$ for $|i-j| \geq 2$.
From the balance equation for state $(i,i-1)$ we obtain $\pi^{(1)}(i,i-1)= \lambda \pi^{0)}(i-1,i-1)$. From the balance equation for state $(i-1,i)$ we obtain $\pi^{(1)}_{i-1,i}= i \mu \pi^{(0)}_{i,i}$. Furthermore, from a cut between speed level $i-1$ and speed level $i$ we obtain $\pi^{(1)}_{i-1,i} = \pi^{(1)}_{i,i-1}$. Combining these three equations leads to the relation $\pi^{(0)}_{i,i}= (\lambda/i \mu) \pi^{(0)}_{i-1,i-1}$ which together with the normalization equation $\sum_i \pi^{(0)}_{i,i} =1$ leads to the solution $\pi^{(0)}_{i,i} = (\rho^i/i!) e^{-\rho}$ and hence to the result that
$\lim_{\nu \to \infty} P(x,y) = e^{\rho(xy-1)}$.
\end{proof}

\subsection{Asymptotics for $\nu \to 0$}
\label{ssec:asymp}
On the other extreme, when the measurements are performed at a much slower rate, time-scale separation between the queue-length and the server-speed occurs. Since the server-speed does not change between two control instants, the queue-length evolves on a faster time-scale compared to the server-speed. In the spatial dimension, both the queue-length and the server-speed can evolve on two scales: fluid scale on which they are $O(\nu^{-1})$ and the normal scale on which they are $O(1)$. The trajectories on these processes rescaled by $\nu$ will be cyclic as shown in Fig. \ref{fig:fluidinfmax}.

\begin{figure}[ht]
    \centering
    \includegraphics[width=0.8\textwidth]{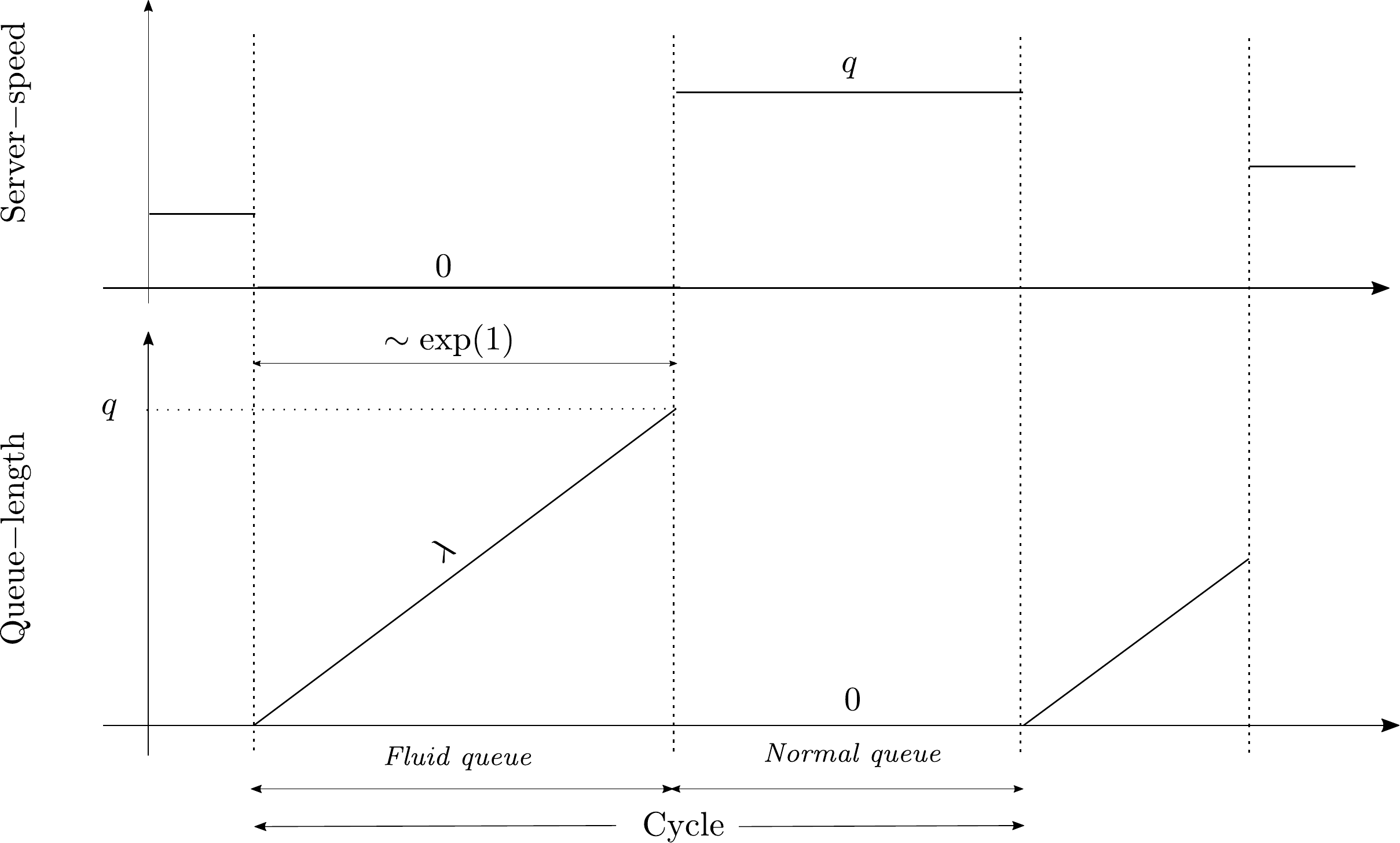}
    \caption{Trajectories of the rescaled queue-length and the rescaled server-speed in the limit $\nu\to 0$. The vertical dashed lines are control instants.}
    \label{fig:fluidinfmax}
\end{figure}

Assume time has been rescaled by a factor $\nu$ so that control instants form a Poisson process of rate $1$. A cycle begins when the measured queue-length is $0$. For this speed, the queue-length process being unstable, it grows linearly on the fluid scale at rate $\lambda$ until the next measurement instant. That is, during this period, the limiting process $\lim_{\nu\to 0} \nu Q(t)$ will grow linearly. The queue-length being on the fluid scale, the next measurement will set the speed to a value $O(\nu^{-1})$. This will bring the queue-length to $0$ instantaneously in time $O(\nu)$. The speed is now $O(\nu^{-1})$ whereas the arrival rate is $\lambda$. So, the queue-length will remain at $0$ (on the normal scale and not just on the fluid scale) until the next measurement instant at which point the speed will be set to $0$. At this point, a new cycle will begin. We were unable to formalize this intuition and leave it a conjecture.

Define $\hP(x,y) :=\lim_{\nu\to 0} P(x^\nu,y^\nu)$ to be the generating function of the scaled pair of random variables $(\nu Q, \nu S)$ in the limit $\nu\to 0$.
\begin{conjecture}
\begin{equation}
	\hP(x,y) = \frac{1}{2} \frac{1}{1 - \lambda \log(x)} + \frac{1}{2} \frac{1}{1 - \lambda \log(y)}.
\end{equation}
\label{prop:nu0}
\end{conjecture}
As $\nu\to 0$, the probability mass of the joint process concentrates around the two axes $Q=0$ and $S=0$. When $Q$ is on the fluid scale, $S=0$ while when $Q=0$, $S$ is on the fluid scale. 

On each of these two axes the scaled stationary process behaves like
an exponentially distributed random variable of rate $\lambda^{-1}$. The coefficient $1/2$ for each of the two terms in the above generating function corresponds to the proportion of time spent by the process on each of the two axes.

\section{Poisson controller with finite maximum speed}
\label{sec:finspeed}
In this section we look at the model with finite maximum speed $\sm$. 
At a control instant, the speed of the server is set equal to either the number of customers observed at that instant, or to $\sm$, whichever is smaller.  As before, between any two consecutive control instants, the speed of the server remains constant at the value chosen at the earlier control instant.

The transition rates of the Markovian process $(Q(t),S(t))_{t\geq 0}$ are now given by
\begin{equation}
(Q(t), S(t)) \to \left\{ \begin{array}{lcl}
			(Q(t) + 1, S(t)) & \mbox{with rate} & \lambda; \\
			(Q(t) - 1, S(t)) & \mbox{with rate} & \mu S(t); \\
			(Q(t), \mbox{min}(Q(t), \sm)) & \mbox{with rate} & \nu. \\
			\end{array}\right.
\end{equation}
which in matrix form can be written as
\begin{equation}
   G= \left(\begin{array}{ccccccc}
     A_{1,0} & A_0 &  &  & & \\
     A_2  & A_{1,1} & A_0 &  &  & \\
      & \ddots & \ddots & \ddots & & \\
      & & A_2 & A_{1,\ell}  & A_0 &   \\
     & & & \ddots & \ddots & \ddots  
     \end{array}
    \right).
\label{eqn:genfin}
\end{equation}
Here, the $(\sm+1) \times (\sm+1)$ matrices $A_0$, $A_{1,\ell}$ and $A_2$ are given by
\[
(A_0)_{i,j} = \left\{ \begin{array}{ll}
			0, & i \neq j; \\
			\lambda, & i = j; \\
			\end{array}\right.,  \quad 
(A_2)_{i,j} = \left\{ \begin{array}{ll}
			0, & i \neq j; \\
			j \mu, & i = j; \\
			\end{array}\right.,
\]
and, for $\ell \geq 0$,
\[
(A_{1,\ell})_{i,j} = \left\{ \begin{array}{ll}
            0, & i \neq j \mbox{ and } j \neq \min(\ell,\sm); \\
            \nu, & i \neq j \mbox{ and } j = \min(\ell,\sm); \\
            -(\lambda + j\mu + \nu), & i = j \mbox{ and } j \neq \min(\ell,\sm); \\
            -(\lambda + \min(\ell,\sm) \mu), & i = j \mbox{ and } j = \min(\ell,\sm); \\
            \end{array}\right.
\]

Figure \ref{fig:transition_diag_finite} shows the rate diagram of the model with $\sm=2$.
 \begin{figure}[!htb]
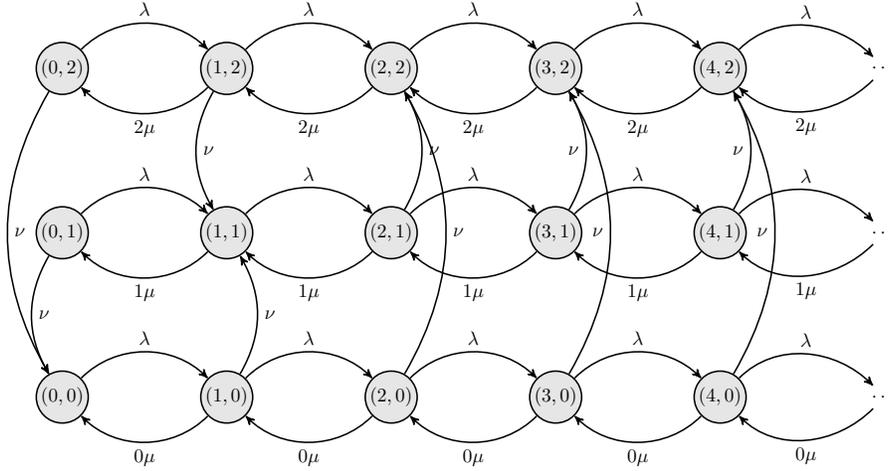

  \centering
 \includestandalone[height=0.5\textwidth]{trans_graph_sm2}
    \caption{Rate diagram of the model with $\sm=2$.}
    \label{fig:transition_diag_finite}
  \end{figure}
 
 Notice that this process $(Q(t),S(t))_{t\geq 0}$ is a Quasi-Birth-Death (QBD) process with level-dependent transition rates. 
 Furthermore notice that for $Q(t) \geq \sm$ the transition rates become level-independent. 
The ergodicity condition of the process $(Q(t),S(t))_{t\geq 0}$ is given by $\lambda<\bar{s}\mu$, for every $\nu>0$. This follows from the fact that the Markov process with rate matrix $A=A_0+A_{1,\bar{s}}+A_2$ is reducible with as only end class state $\bar{s}$. In this state the arrival rate is given by $\lambda$ and the service rate is given by $\bar{s}\mu$ and hence the ergodicity condition is given by the mean drift condition $\lambda<\bar{s}\mu$ 
(see Theorem 7.3.1 in \citet{latouche1999}).

\begin{remark}
    As mentioned in Sec. \ref{sec:infspeed}, in the case that $\bar{s} = \infty$ the system will be ergodic for any combination of $\lambda>0$, $\mu>0$ and $\nu>0$. This can be shown for example in the following way. Assume that $\lambda$ and $\mu$ are such that $(\tilde{s}-1)\mu \leq \lambda < \tilde{s} \mu$. Furthermore, take $U=\{(i,j) : 0 \leq i \leq \tilde{s}-1, 0 \leq j \leq \tilde{s}-1\}$, the set of states in which both the number of customers and the service speed are at most equal to $\tilde{s}-1$. Now the expected duration of an arbitrary excursion from the set $U$ in the model with maximal speed $\bar{s} = \infty$ will be smaller than the expected duration of the corresponding excursion in the model with maximal speed $\tilde{s} < \infty$ (essentially because during excursions outside the set $U$ the server in the model with infinite maximal speed $\bar{s} = \infty$ works always at least as fast as in the model with finite maximal speed $\tilde{s} < \infty$). As in the latter model this expected duration will be finite, because $\lambda < \tilde{s} \mu$, the expected duration will also be finite in the model with maximal speed $\bar{s} = \infty$ and hence the system will be ergodic.
    \label{rem:ergodic}
 \end{remark}

 The vector $\pi_n = [\pi_{n,0},\pi_{n,1}, \ldots, \pi_{n,\sm}]$ containing the stationary probabilities for the different states with $n$ customers in the system, are, for $n \geq \sm-1$, of the form
 \begin{equation}
 \pi_n = \pi_{\sm-1} R^{n-\sm+1},
 \label{eqn:qbd_gen}
 \end{equation}
where $R$ is the minimal non-negative solution of the matrix equation 
\[
A_0+RA_1+R^2 A_2=0,
\]
where $A_1 = A_{1,\sm}$.

The matrix $R$ can be explicitly calculated for this model. It is a diagonal matrix with additional non-zero values in the last column. 
\begin{proposition}
For the matrix $R$ we have that
\begin{equation}
    R = \left(\begin{array}{cccccc}
     \frac{\lambda}{\lambda+\nu} &  &  &  &  &\frac{\lambda}{\sm\mu} \\
       & \frac{\lambda \beta_1(\nu)}{\mu} &  & &  & \frac{\lambda(1-\beta_1(\nu))}{\sm\mu} \\
      & & \ddots &  &  & \vdots \\
       &  &  & \frac{\lambda \beta_{i}(\nu)}{i\mu} &  &\frac{\lambda(1-\beta_{i}(\nu))}{\sm\mu} \\
     &  &  & & \ddots & \vdots \\
      &  &  &  &  & \frac{\lambda}{\sm\mu}  
     \end{array}
    \right).
\end{equation}
\label{prop:Rmat}
\end{proposition}
\begin{proof}
We use the fact that $R=A_0 N = \lambda I N = \lambda N$, where the matrix $N$ contains, for an arbitrary $n \geq \sm$, the expected sojourn times in the different states on level $n$, starting from the different states on level $n$, before the first visit to level $n-1$ (see Theorem 6.4.1 in \citet{latouche1999}). Now, clearly $N_{i,j}=0$ for all $i \neq j \mbox{ and } j \neq \sm$ because for these values of $i$ and $j$ you can only reach state $(n,j)$ starting from state $(n,i)$ via level $n-1$. Furthermore, $N_{0,0} = \tfrac{1}{\lambda+\nu}$, because if you start in state $(n,0)$ you stay there an exponentially distributed time with parameter $\lambda + \nu$ and will never return there before the first visit to level $n-1$. The value $N_{\sm,\sm}$ is equal to the expected time there is one customer in the system during a busy period of an $M/M/1$ queue with arrival rate $\lambda$ and service rate $\sm\mu$. Hence, $N_{\sm,\sm}= \frac{1}{\lambda + \sm\mu} + \frac{\lambda}{\lambda+\sm\mu} N_{\sm,\sm}$ and so $N_{\sm,\sm}=\frac{1}{\sm\mu}$. Furthermore, we have that $N_{0,\sm}=N_{\sm,\sm}$ because starting from state $(n,0)$ we certainly reach state $(n,\sm)$ before the first visit to level $n-1$. For $i \neq 0 \mbox{ and } i \neq \sm$ we have that $N_{i,\sm}=(1-\beta_i(\nu))N_{\sm,\sm}$, because $\beta_i(\nu)$ is the probability that the busy period in an $M/M/1$ queue with arrival rate $\lambda$ and service rate $i\mu$ is smaller than an independent exponential random variable with parameter $\nu$. Hence $1-\beta_i(\nu)$ is the probability that we reach state $(n,\sm)$ before the first visit to level $n-1$ when we start in a state $(n,i)$
with $i \neq 0 \mbox{ and } i \neq \sm$.

Finally, the quantity $N_{j,j}$ for $j \neq 0 \mbox{ and } j \neq \sm$ is equal to the expected time there is one customer in the system during a busy period of an $M/M/1$ queue with arrival rate $\lambda$ and service rate $j\mu$ and in which a disaster occurs removing all the customers in the system after an exponentially distributed time with parameter $\nu$. We have that $N_{j,j}= \frac{1}{\lambda + \nu + j\mu} + \frac{\lambda}{\lambda+\nu+j\mu} \beta_j(\nu) N_{j,j}$ and so $N_{j,j}=\frac{1}{\lambda(1-\beta_j(\nu))+\nu+j\mu} = \frac{\beta_j(\nu)}{j\mu}$. 
\end{proof}
Since $R$ is of the above form, $R^n$ has the following simple formula which can be used for computing $\pi_n$ in \eqref{eqn:qbd_gen}.
\begin{corollary}
\label{cor:rn}
Let $R_{i,j}$ be the element (i,j) of $R$. Then,
\begin{equation}
    R^n = \left(\begin{array}{ccccc}
    R_{0,0}^n &  &  &  & R_{0,\sm}\frac{\left(R_{\sm,\sm}^n - R_{0,0}^n\right)}{R_{\sm,\sm}-R_{0,0}} \\
      & \ddots &  &  & \vdots \\
       &  &  R_{i,i}^n &  & R_{i,\sm}\frac{\left(R_{\sm,\sm}^n - R_{i,i}^n\right)}{R_{\sm,\sm}-R_{i,i}}\\
      &  &  & \ddots &  \vdots \\
      &  &  &  & R_{\sm,\sm}^n
     \end{array}
    \right).
\end{equation}
\end{corollary}
From the above corollary, it can inferred that the joint probability vector of level $n \geq \sm$ and speeds other than $\sm$ is a geometric term with rate $R_{i,i}$. For level $n\geq\sm$ and speed $\sm$, the probability is a mixture of geometric terms:  terms with rate $R_{i,i}$, $i < \sm$, and one with rate $R_{\sm,\sm}$.

The probability vectors $\pi_{n}$, $n\leq\sm-1$, can now be computed using a system of linear equations and the normalization equation (see, e.g., \citet{latouche1999}). It then follows that the joint generating function can be expressed as
\begin{proposition}
\label{prop:gf_fin}
\begin{equation}
    \label{eqn:gf_fin}
    P(x,y) = \sum_{n < \sm-1, j} \pi_{n,j}x^n y^j + \pi_{\sm-1}(I-Rx)^{-1}x^{\sm-1}\cdot\mathbf{y}^T,
\end{equation}
with $\mathbf{y} = [1, y, y^2, \hdots, y^{\sm}]$
\end{proposition}

Observe that from the particular form of $R$ in Prop. \ref{prop:Rmat}, $(I-Rx)^{-1}$ can be explicitly computed:
\begin{equation}
    (I-Rx)^{-1} = \left(\begin{array}{ccccc}
    \frac{1}{1-R_{0,0}x} &  &  &  &\frac{R_{0,\sm}x}{(1-R_{\sm,\sm}x)(1-R_{0,0}x)} \\
      & \ddots &  &  & \vdots \\
       &  &  \frac{1}{1-R_{i,i}x} &  & \frac{R_{i,\sm}x}{(1-R_{\sm,\sm}x)(1-R_{i,i}x)}\\
      &  &  & \ddots & \vdots \\
      &  &  &  & \frac{1}{1-R_{\sm,\sm}x}
     \end{array}
    \right).
\end{equation}

As an illustration of this analysis, we obtain the steady-state probabilities for $\sm=1$. In this case, our model is a special case of the model with two different service speeds discussed in \citet{B2008}. 

For $\sm=1$, \eqref{eqn:qbd_gen} and Prop. \ref{prop:Rmat} reduce to
\begin{equation}
[\pi_{n,0}, \pi_{n,1}] = [\pi_{0,0}, \pi_{0,1}] R^n,
\end{equation}
with
\begin{equation}
R = \left(\begin{array}{cc}
     \frac{\lambda}{\lambda+\nu} & \frac{\lambda}{\mu} \\
     0                           & \frac{\lambda}{\mu}
     \end{array}
    \right).
\end{equation}
Furthermore,
\begin{equation}
\pi_{0,1} = \tfrac{\lambda}{\nu} \pi_{0,0}
\end{equation}
and
\begin{equation}
R^n = \left(\begin{array}{cc}
     \left(\frac{\lambda}{\lambda+\nu}\right)^n & \frac{\lambda+\nu}{\lambda+\nu-\mu}\left[\left(\frac{\lambda}{\mu}\right)^n - \left(\frac{\lambda}{\lambda+\nu}\right)^n\right] \\
     0                           & \left(\frac{\lambda}{\mu}\right)^n
     \end{array}
    \right).
\end{equation}
Hence, we obtain
\begin{eqnarray}
\pi_{n,0} &=& \left(\tfrac{\lambda}{\lambda+\nu}\right)^n \pi_{0,0}, \\
\pi_{n,1} &=& \left[ \left(\tfrac{\lambda}{\nu}+  \tfrac{\lambda+\nu}{\lambda+\nu-\mu}\right) \left(\tfrac{\lambda}{\mu}\right)^n -  \tfrac{\lambda+\nu}{\lambda+\nu-\mu} \left(\tfrac{\lambda}{\lambda+\nu}\right)^n \right] \pi_{0,0}.
\end{eqnarray}
The constant $\pi_{0,0}$ follows of course from the normalization equation, and takes the value
\begin{equation}
\pi_{0,0} = \tfrac{\nu (\mu - \lambda)}{\mu (2\lambda + \nu)}.
\end{equation}
Remark that
\begin{equation}
P(S=0) = \tfrac{(\mu-\lambda)(\lambda+\nu)}{\mu(2\lambda+\nu)}, \quad
P(S=1) = \tfrac{\lambda(\lambda+\mu+\nu)}{\mu(2\lambda+\nu)}
\end{equation}
and that
\begin{equation}
P(Q=0,S=1) = \pi_{0,1} = \tfrac{\lambda (\mu - \lambda)}{\mu (2\lambda + \nu)}
\end{equation}
so that
\begin{equation}
P(Q>0,S=1) = \tfrac{\lambda(\lambda+\mu+\nu)- \lambda (\mu - \lambda)}{\mu(2\lambda+\nu)}
= \tfrac{\lambda(2\lambda+\nu)}{\mu(2\lambda+\nu)}= \tfrac{\lambda}{\mu}
\end{equation}
as it should be. 

For the joint probability generating function we have
\begin{eqnarray}
P(x,y) &=& \sum_{n=0}^{\infty} \pi_{n,0} x^n + \sum_{n=0}^{\infty} \pi_{n,1} x^n y \\
&=& \pi_{0,0} \left[ \tfrac{\lambda+\nu}{\nu+\lambda(1-x)} - \tfrac{\lambda+\nu}{\lambda+\nu-\mu} \cdot \tfrac{\lambda+\nu}{\nu+\lambda(1-x)} \cdot y \right. \nonumber\\
    && \quad  \left. +\left(\tfrac{\lambda}{\nu}+  \tfrac{\lambda+\nu}{\lambda+\nu-\mu}\right) \cdot \tfrac{\mu}{\mu-\lambda x} \cdot y\right]  \label{eqn:gf_s1}
\end{eqnarray}

\subsection{Asymptotics for $\nu\to \infty$}
\label{ssec:fin_inf}
In this asymptotic regime, intuitively, as in Sec.~\ref{sec:infspeed}, the joint probability will be non-zero only on states of the type $(q,q)$ for $q<\sm-1$ and $(q,\sm)$ for $q\geq\sm$. The stationary probability of seeing queue-length $q$ will be that of an $M/M/\sm$ queue. Formally, 
\begin{proposition}
\label{prop:fininf}
\begin{equation}
    \lim_{\nu\to \infty} \pi_{q,j} = \begin{cases}
                                \pi_{0,0}\frac{\rho^q}{j! \sm^{(q-j)}}, & j=\min(q,\sm);\\
                                0, & \mbox{otherwise}.
                                \end{cases}
\end{equation}
\end{proposition}
The proof is based on arguments from singular perturbation theory \citep{AAN04}. Here, we give a sketch of the proof. Rescale time by $\nu$ so that the control rate is $1$, the arrival rate is  $\lambda\nu^{-1}$ and the service requirement has rate $\mu\nu^{-1}$. With this transformation, the generator in \eqref{eqn:genfin} can be written as
\begin{equation}
    G = G_0 + \epsilon G_1
\end{equation}
with $\epsilon=\nu^{-1}$. Note that $G_0$ is a block diagonal matrix. That is, upon setting $\epsilon = 0$, each block becomes a separate ergodic class. Thus, this chain is singularly perturbed for $\epsilon > 0$. The limiting distribution for $\epsilon\to 0$ can be computed from the stationary distribution of each block of $G_0$ as well as that of the aggregated chain (see Sec. 4 in \citet{AAN04}). It can be easily checked that the stationary distribution of block $q$ of $G_0$ is $\mathbf{e}_{\min(q,\sm)}$, where $\mathbf{e}_i$ is the unit vector with $1$ in the $i$th column. This is true since any observation in block $q$ sets the speed to $\min(q,\sm)$. Moreover, the aggregated chain has the dynamics of an $M/M/\sm$ queue with arrival rate $\lambda$ and service requirement $\mu^{-1}$. Combining these two observations, the proof follows.

\subsection{Asymptotics for $\nu\to 0$}
We shall use similar arguments as in Sec.~\ref{ssec:asymp} to obtain the generating function for the finite maximum speed case as $\nu\to 0$. In the spatial dimension, the queue-length process can be either on the fluid scale or on the normal scale depending on whether the speed of the server is smaller or larger than $\lambda$. Unlike in the infinite maximum speed case, here the speed always remains on the normal scale.

In the following, we shall assume that time has been rescaled by a factor $\nu$ so that the control instants happen according to a Poisson process of rate $1$. In order to derive various quantities related to the queue-length and the speed, we shall use an intuitive argument based on time-scale separation to first compute the marginal distribution of the server-speed. From the marginal distribution, the generating function of the joint process will be obtained using \eqref{eqn:altP}.

Define 
\begin{equation}
	\S^+ = \{ j : j > \rho\}, \quad
	\S^- = \{ j : j \leq \rho\}.
\end{equation}
The set $\S^+$ contains the speeds for which the queue is positive recurrent while $\S^-$ is its complement set. 

The marginal distribution of the speed will be computed from the stationary distribution of the Markov chain embedded at the control instants (and one more point, which will be explained later) and the renewal reward theorem. Let a cycle denote the time between two consecutive observations in $\S^{-}$. An illustration of trajectories of the {\it queue-length on the fluid scale} and the speed is shown in Fig. \ref{fig:fluidmax}. Until the first control instant after the start of a cycle, the queue-length grows linearly on the fluid scale with rate $\lambda -j\mu$ where $j$ is the queue-length sampled at the start of the cycle. Hence, at the first control instant the speed will necessarily be set to $\sm$, which is in $\S^+$. Since the queue-length process is now stable, it will decrease with rate $\sm\mu - \lambda$ for one or more control instants until the fluid hits $0$. (Notice that control instants do not affect the server working at constant speed $\sm$ while the queue is at the fluid level, ensuring that level 0 will be reached.) The queue-length process will now evolve on the normal scale for one or more control instants until a speed from $\S^-$ is sampled. At this point, a new cycle will begin. A cycle can thus be decomposed into three phases: the \fu\ phase during which the fluid grows; the \fs\ phase during which the fluid drains; and the \no\ phase during which the queue length lives on the normal scale.

\begin{figure}
    \centering
    \includegraphics[width=0.8\textwidth]{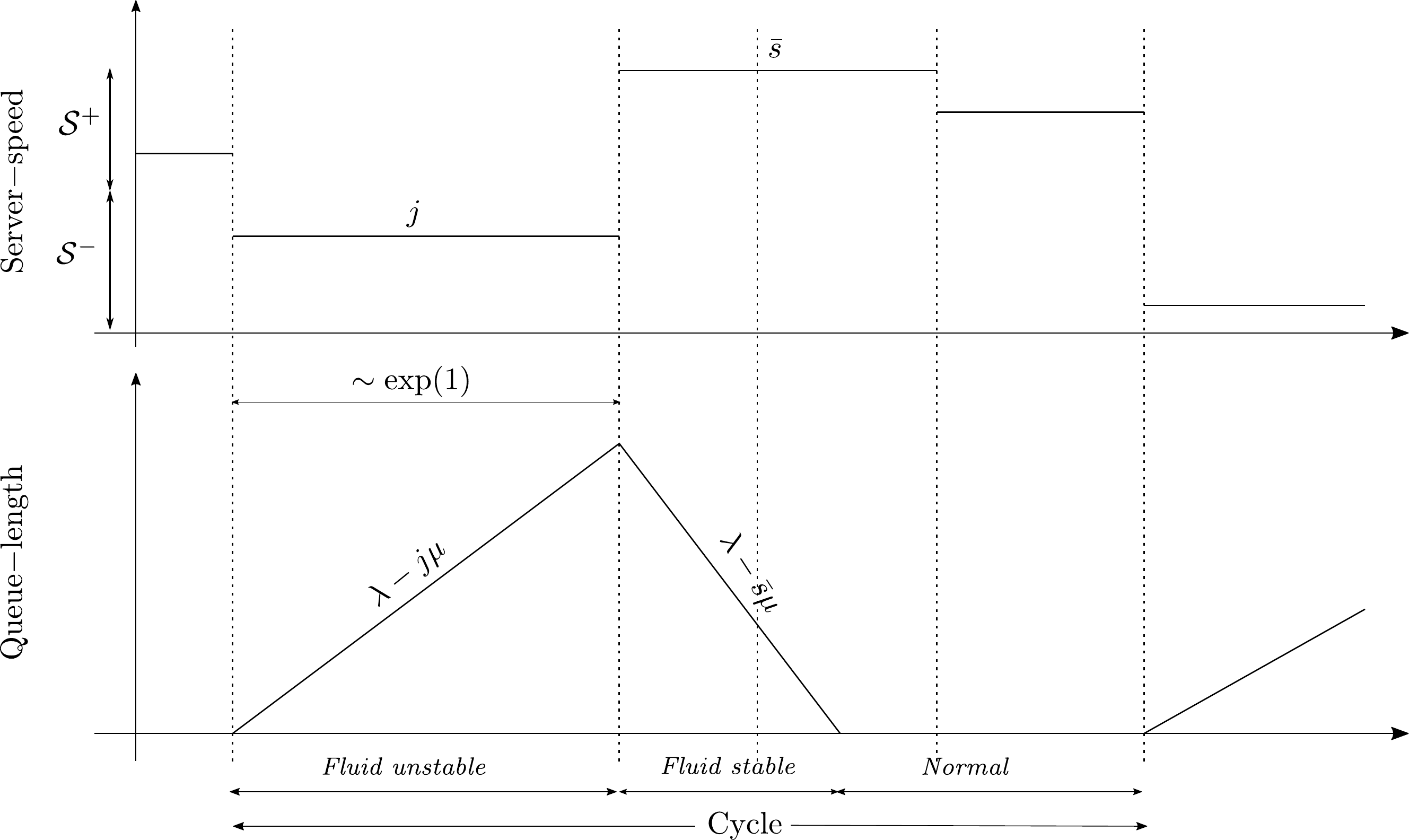}
    \caption{Trajectories of the queue-length on the fluid scale and the server speed in the limit $\nu\to 0$. The vertical dashed lines are control instants. Trajectories that occur on the normal scale all collapse to 0.}
    \label{fig:fluidmax}
\end{figure}

Let $\hat{S}_n$ be the server-speed process embedded just after control instants except in the \textit{fluid stable} state when the process is embedded at instants when the fluid hits $0$. Observe that $\sm$ is the only state in which the queue-length process can be in both the fluid scale as well as the normal scale. For the computation of the generating function, it will be convenient to compute the marginal probability of the server-speed being $\sm$ separately for each of the two scales of the queue-length process. For this,  define the server-speed $\sm_f$ which indicates the maximum speed when the queue-length is on the fluid scale. We shall use $\sm$ as before for the maximum speed when the queue-length is on normal scale. $\hat{S}_n$ is a discrete-time Markov chain taking values in $\S\cup\{\sm_f\}$, with transition probability matrix
\begin{align}
M =
 \begin{blockarray}{cccc}
 & \S^- & \S^{+} & \sm_f\\
 \begin{block}{c[ccc]}
 \S^- & \mathbf{0} & \mathbf{0} & \mathbf{1}\\
 \S^+ & V_0 & V & \mathbf{0}\\
 \sm_f & \mathbf{0} & \mathbf{b} & 0\\
 \end{block}
 \end{blockarray}.	
\end{align}
The elements $(i,j)$ of matrix $M$ for $i\in\S^+$ and $j\in\S^-\cup\S^+$ (these are all entries of the matrices $V_0$ and $V$) can be deduced directly: $M_{i,j}$ is the probability of setting the speed to $j\in\S^-\cup\S^+$ at the next control instant given that the current speed is $i \in \S^+$. Because of time-scale separation, between two control instants in the normal phase, given the server-speed to be $i\in\S^+$, the queue-length process is a stationary $M/M/1$ queue with arrival rate $\lambda$ and service rate $i\mu$. Thus,
\begin{equation}
    M_{i,j} = 
    \begin{cases}
    (1-\rho_i)\rho_i^j, & j \neq \sm;\\
    \rho_i^{\sm}, & j=\sm.
    \end{cases}
\end{equation}
with $\rho_i = \lambda/(i\mu)$. The vector $\mathbf{b} = [0,0,\hdots, 0, 1]$ since from $\sm_f$ the chain can only go to $\sm$ (at the additional embedded time instant when the fluid hits 0).

Let $\psi$ be the stationary distribution of $\hat{S}_n$. Then, by the renewal reward theorem, the marginal distribution of the speed in the limit $\nu\to 0$ can be written in terms of $\psi$ as follows: 
\begin{equation}
    \sigma_j = \frac{\tau_j\psi_j}{\sum_{k\in\S\cup\{\sm_f\}} \tau_k\psi_k}, \; j \in \S\cup\{\sm_f\},
    \label{eqn:sig_finm}
\end{equation}
where $\tau_j$ is the average time spent in state $j$ before the next jump. We remind the reader that the state $\sm$ has been split in two: $\sm_f$ refers to the maximum speed when the queue-length is on the fluid scale while $\sm$ refers to maximum speed when the queue-length is on the normal scale.  Note that for all states except $\sm_f$, the jumps occur after an exponentially distributed time of rate $1$. The only unknown quantity is thus $\tau_{\sm_f}$, which can be written as
\begin{equation}
    \tau_{\sm_f} = \frac{\sum_{j\in\S^{-}}\tau_{\sm_f,j}\psi_j}{\sum_{j\in\S^-}\psi_j}.
\end{equation}
where $\tau_{\sm_f,j}$ is the expected time spent in $\sm_f$ conditioned on the speed being $j$ in the preceding \fu\ phase. It can be seen that, given speed $j$ in the \fu\ phase, the amount of fluid at the end of this phase is distributed as $\exp((\lambda-j\mu)^{-1})$. The expected amount of fluid at the start of \fs\ phase is thus $\lambda - j\mu$. The speed in the \fs\ phase being $\sm$, the expected time it takes to drain this fluid is 
\begin{equation}
    \tau_{\sm_f,j} = \frac{\lambda - j\mu}{\sm\mu - \lambda}.
\end{equation}

Let us decompose the vector $\psi$ in the same manner as the matrix $M$: $\psi=[\psi^-,\psi^+,\psi_{\sm_f}]$. The component $\psi^{-}$ thus contains the stationary probabilities of the states in $\S^-$. The other components are similarly defined. The following result presents formulas for $\psi$ and $\tau_{\sm_f}$. Its proof is by a simple check and is therefore  omitted.
   
\begin{proposition}
Let $\theta = \mathbf{b}(I-V)^{-1}$ and $\kappa = \theta\cdot\mathbf{1}^T$. Then,
\begin{align}
    \psi^- = \frac{\theta V_0}{2+\kappa},\quad \psi^+ = \frac{\theta}{2 + \kappa}, \quad \bar{\psi}_{\sm_f} = \frac{1}{2+\kappa},
    \label{eqn:psi}
\end{align}
and
\begin{equation}
\label{eqn:tausmf}
    \tau_{\sm_f} = \frac{1}{\sm\mu-\lambda}\left(\lambda - \theta V_0 J^T\mu\right)
\end{equation}
with $J=[0,1,\hdots,\lvert \S^-\rvert]$.
\end{proposition}
Intuitively, $\theta$ computes the number of visits of $\hat{S}_n$ to states in $\S^+$ starting in state $\sm_f$ while remaining in $\S^+$.  Upon leaving $\S^+$, $\hat{S}_n$  enters states in $\S^-$ with probability distribution $\theta V_0$, initiating a \fu\ phase. It leaves $\S^-$ after one epoch (corresponding to the first next control instant) and visits $\sm_f$ where it stays for one epoch before entering $\S^+$ through the state $\sm$ (which corresponds to the fluid hitting 0).

We now have all the ingredients to compute the generating function $P(x,y)$. Define $\hat{P}(x,y) = \lim_{\nu\to 0} P(x^\nu, y)$ for the generating function that captures the queue-length on the
fluid scale and $\tilde{P}(x,y) = \lim_{\nu\to 0} P(x,y)$ to be the one that captures the normal scale queue-length. 
From \eqref{eqn:altP} and \eqref{eqn:sig_finm}, we get
\begin{align}
    \hat{P}(x,y) &= \sum_{j\in\S^-}\frac{\sigma_j}{1-(\lambda-j\mu)\log(x)}\left(y^j + \frac{\lambda-j\mu}{\sm\mu-\lambda}y^{\sm}\right)+ \sum_{j\in\S^+}y^j\sigma_j,\label{eqn:hP}\\
    \tilde{P}(x,y) &=\sum_{j\in\S^+}\frac{1-\rho_j}{1-\rho_jx}y^j\sigma_j.
    \label{eqn:tP}
\end{align}
On the fluid scale, conditional on speed being $j\in\S^-$, the queue-length is an exponential random variable of rate $\lambda-j\mu$. This is also true in the \text{fluid stable} phase. The main difference between the two phases is the time spent in there. Since for a \textit{fluid unstable} phase of length $1$, the corresponding length of the \textit{fluid stable} phase is $\tau_{\sm_f}$, we get the terms for $j\in\S^-$ in \eqref{eqn:hP}. The terms for $\S^+$ are a direct consequence of the fact that on the fluid scale, a stable queue is always of length $0$ and only the marginal distribution of the speeds appears. On the normal scale, the queue-length is just a stationary $M/M/1$ queue which gives the coefficient of $y^j\sigma_j$ in \eqref{eqn:tP}.

\begin{remark}
The above asymptotic analysis cannot be used to obtain the results in Sec. \ref{ssec:asymp} for the infinite speed case by taking $\sm \to \infty$. The difficulty comes from the fact that when $\sm = \infty$, the speed sampled after the fluid (unstable) phase is an exponential random variable. On the other hand, in the finite $\sm$ case, this speed is always a fixed value. The limit of this sequence of deterministic values is unable to capture the exponential random variable at $\sm=\infty$. One will have to scale $\sm$ as $\nu^{-1}$ to get back the results of Sec. \ref{ssec:asymp}.  
\end{remark}

Next we apply the above arguments to the special case $\sm=1$ and show that we indeed obtain the same expression for the generating function as the one in \eqref{eqn:gf_s1} with $\nu\to 0$.

For $\sm=1$, the transition matrix of the embedded chain becomes
\begin{align}
M =
 \begin{blockarray}{cccc}
 & 0 & 1 & 1_f\\
 \begin{block}{c[ccc]}
 0 & 0 & 0 & 1\\
 1 & 1-\rho & \rho & 0\\
 1_f & 0 & 1 & 0\\
 \end{block}
 \end{blockarray},		
\end{align}
for which $\psi = (1+2(1-\rho))^{-1}[1-\rho, 1, 1-\rho]$, and $\tau_{1_f} = \rho(1-\rho)^{-1}$. The marginal distribution of the speed is thus
\begin{equation}
    \sigma = \left[\frac{1-\rho}{2}, \frac{1}{2}, \frac{\rho}{2}\right].
    \label{eqn:sigma_s1}
\end{equation}

We now check that the above values are consistent with \eqref{eqn:gf_s1} with $\nu\to 0$. Setting $x$ to $x^\nu$ and taking the limit $\nu\to 0$ in \eqref{eqn:gf_s1}, we get
\begin{equation}
\hP(x,y) = \tfrac{1-\rho}{2} \cdot \tfrac{1}{1-\lambda \log(x)} + \tfrac{\rho}{2} \cdot \tfrac{1}{1-\lambda \log(x)} \cdot y + \tfrac12 \cdot y,
\end{equation}
which is the same expression as the one when \eqref{eqn:sigma_s1} is substituted in \eqref{eqn:hP}.

The queue length process is with probability $\tfrac{1-\rho}{2}$ on the fluid scale with speed $0$, it is with probability $\tfrac{\rho}{2}$ on the fluid scale with speed $1$, and it is with probability $\tfrac12$ on the normal scale with speed 1. Furthermore, remark that when the queue length process is on the normal scale, the number of customers in the system is geometrically distributed with parameter $\rho$ and hence the probability that the server is working on this scale is $\rho$. In total, the fraction of time the server is working is $\tfrac{1-\rho}{2} \cdot 0 + \tfrac{\rho}{2} \cdot 1 + \tfrac12 \cdot \rho = \rho$ as it should be.

\section{Other models with transitions to diagonal states}
In this section, we analyze two models where the controller is just an observer and does not modify the speed of the server. 
\subsection{The $M/M/1$ queue with a Poisson observer}
Consider the $M/M/1$ queue with arrival rate $\lambda$, exponentially distributed service times with parameter $\mu$ and with an observer
arriving to the system according to a Poisson process with
rate $\nu$. The server always works with speed 1 (so no speed adaptations will occur) and assume $\lambda < \mu$.

In this case, the  process $(Q(t),S(t))_{t\geq 0}$ is a Markov process with transitions
\begin{equation}
(Q(t), S(t)) \to \left\{ \begin{array}{lcl}
			(Q(t) + 1, S(t)) & \mbox{with rate} & \lambda; \\
			(Q(t) - 1, S(t)) & \mbox{with rate} & \mu; \\
			(Q(t), Q(t)) & \mbox{with rate} & \nu. \\
			\end{array}\right.
\end{equation}
Figure \ref{fig:transition_diag_3} shows the rate diagram of the $M/M/1$ queue with Poisson observer.
 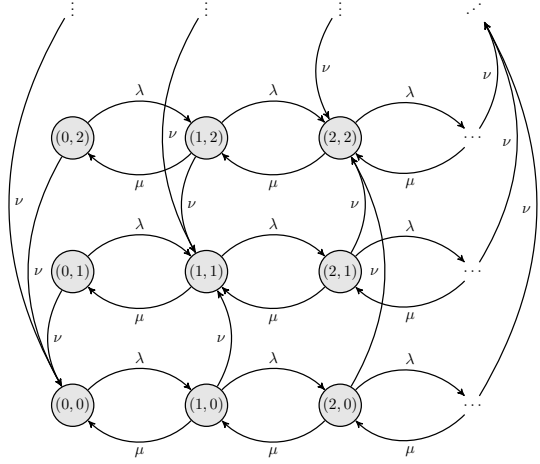
\begin{figure}[!htb]
  \centering
  \tikzstyle{input}=[font=\small\sffamily\bfseries]
\tikzstyle{rect}=[rectangle, draw=black, font=\small\sffamily\bfseries, inner sep=9pt]
\begin{tikzpicture}[->,>=stealth',auto,node distance=3cm,
  thick,main node/.style={fill=gray!20,draw, circle,inner sep=1pt}, ph node/.style={draw=none}]

  \foreach \x in {0,...,3}
    \foreach \y in {0,...,3} {
        \ifthenelse{\x=3 \OR \y=3}{
        		\ifthenelse{\x=3 \AND \y=3}{\node[draw=none] (\x\y) at (3*\x,3*\y) {\reflectbox{$\ddots$}};}{
		 	\ifthenelse{\x = 3}{\node[draw=none] (\x\y) at (3*\x,3*\y){$\cdots$};}{\node[draw=none] (\x\y) at (3*\x,3*\y){$\vdots$};};};
			}{\node [main node]  (\x\y) at (3*\x,3*\y) {$(\x, \y)$};};
       };

%

\foreach \x [count=\xi] in {0,...,2}
    \foreach \y [count=\yi] in {0,...,2}
      {\draw [->] (\x\y) to [out=45,in=135] node [pos=0.5,above]  {$\lambda$} (\xi\y);
       \draw [->] (\xi\y) to [out=225,in=315] node [pos=0.5, below] {$\mu$}(\x\y);
       };
\foreach \x in {0,...,3}
        \foreach \y in {0,...,3}{
		 \ifthenelse{\x=\y}{}{\ifthenelse{\x > \y}{\draw [->] (\x\y) to [bend right] node [pos=0.5, left] {$\nu$}(\x\x);}{\draw [->] (\x\y) to [bend right] node [pos=0.5, right] {$\nu$}(\x\x);};};
	}

\end{tikzpicture}
    \caption{Rate diagram of the $M/M/1$ queue with Poisson observer.}
    \label{fig:transition_diag_3}
  \end{figure}

The balance equations for the stationary probabilities are
\begin{align}
	(\lambda + \mu + \nu)\pi_{i,j} &= \lambda\pi_{i-1,j} + \mu \pi_{i+1,j} + \1_{\{j=i\}}\nu\sum_{k}\pi_{i,k},  \, i \geq 1, j\geq 0; \nonumber\\
	(\lambda + \nu)\pi_{0,j} &= \mu \pi_{1,j} + \1_{\{j=0\}}\nu\sum_{k}\pi_{0,k},  \, j\geq 0, \nonumber
\end{align}

\begin{proposition}
Generating function $P(x,y)$ satisfies the functional equation
\begin{equation}
	\left[\nu + \lambda(1-x)+ \mu\left(1-\tfrac{1}{x}\right)\right]P(x,y) =  \mu \left(1-\tfrac{1}{x}\right) P(0,y) + \nu P(xy,1).
\label{eqn:gen_fun_3}
\end{equation}

The solution of functional equation (\ref{eqn:gen_fun_3}) is given by
\begin{equation}
P(x,y)= \frac{\frac{1-\tfrac{1}{x}}{\tfrac{1}{x_1}-1}\nu \frac{\mu-\lambda}{\mu - \lambda x_1 y}+ \nu \frac{\mu-\lambda}{\mu - \lambda x y}}{\nu + \lambda(1-x)+ \mu\left(1-\tfrac{1}{x}\right)}
\label{eqn:sol_gen_fun_3}
\end{equation}
with $0 < x_1 < 1$ solution of the equation
$\lambda x^2 - \left(\lambda + \mu + \nu\right)x+ \mu=0$.
\label{prop:sol_gen_fun_3}
\end{proposition}
\begin{proof}
Multiplying both sides of the rate balance equations by $x^iy^j$ and summing
over all possible $i$ and $j$ immediately leads to the equation
\begin{eqnarray*}
&&	(\nu + \lambda)P(x,y) + \mu \left[P(x,y) - P(0,y)\right] \\
&&= \lambda x P(x,y) + \frac{\mu}{x} \left[P(x,y) - P(0,y)\right] + \nu P(xy,1),
\end{eqnarray*}
which can be alternatively written as (\ref{eqn:gen_fun_3}).
The proof of (\ref{eqn:sol_gen_fun_3}) makes use of the fact that
$P(xy,1) = \tfrac{\mu-\lambda}{\mu - \lambda x y}$ and the fact that if
$0 < x_1 < 1$ is a solution of the equation
$\lambda x^2 - \left(\lambda + \mu + \nu\right)x+ \mu=0$ then
\[
\mu \left(1-\tfrac{1}{x_1}\right) P(0,y) + \nu \tfrac{\mu-\lambda}{\mu - \lambda x_1 y} =0.
\]
\end{proof}

\begin{corollary}[Asymptotics when $\nu \to \infty$]
From (\ref{eqn:sol_gen_fun_3}) and the fact that $x_1 \to 0$ if $\nu \to \infty$, it immediately follows that $P(x,y) \to \tfrac{\mu-\lambda}{\mu - \lambda x y}$ when $\nu\to \infty$.
\end{corollary}

\begin{corollary}[Asymptotics when $\nu \to 0$]
From (\ref{eqn:sol_gen_fun_3}) and the fact that $x_1 \to 1$ and $\nu /(\tfrac{1}{x_1}-1) \to \mu-\lambda$, if $\nu \to 0$, it immediately follows that $P(x,y) \to \tfrac{\mu-\lambda}{\mu - \lambda x}\cdot \tfrac{\mu-\lambda}{\mu - \lambda y}$ when $\nu \to 0$.
\end{corollary}
The above results can be explained intuitively by the fact that the observer does not influence the queue-length.  In the limit $\nu\to 0$, it just samples from the stationary distribution of an $M/M/1$ queue. Hence, in stationarity, the $S$ and $Q$ processes are distributionally equivalent to two independent $M/M/1$ queues.

\begin{remark}
The joint steady-state distribution of $(Q_t,S_t)$ is equal to the joint steady-state distribution of $(Q_{\tau_t+Y},Q_{\tau_t})$, where the random time $Y=t-\tau_t$ is exponentially distributed with parameter $\nu$. So, $P(x,y)$ is the probability generating function of the joint steady-state distribution of the queue length in an $M/M/1$ queue at two different time points, where the distance between the two time points is exponentially distributed with parameter $\nu$. Formula  (\ref{eqn:sol_gen_fun_3}) can be alternatively derived using results on the transient distribution in the $M/M/1$ queue (see Section I.4.4 and in particular formula (4.27) in \citet{cohen1982}).
\end{remark}

 \subsection{The $M/M/\infty$ queue with an observer}
Next, we consider the infinite server model with an observer arriving to the system according to a Poisson process with
rate $\nu$. As before, customers arrive according to a Poisson process with rate $\lambda$ and they require exponentially
distributed service times with parameter $\mu$.

The  process $(Q(t),S(t))_{t\geq 0}$ is a Markov process with transitions
\begin{equation}
(Q(t), S(t)) \to \left\{ \begin{array}{lcl}
			(Q(t) + 1, S(t)) & \mbox{with rate} & \lambda; \\
			(Q(t) - 1, S(t)) & \mbox{with rate} & \mu Q(t); \\
			(Q(t), Q(t)) & \mbox{with rate} & \nu. \\
			\end{array}\right.
\end{equation}

Figure \ref{fig:transition_diag_2} shows the rate diagram of the $M/M/\infty$ queue with Poisson observer.
 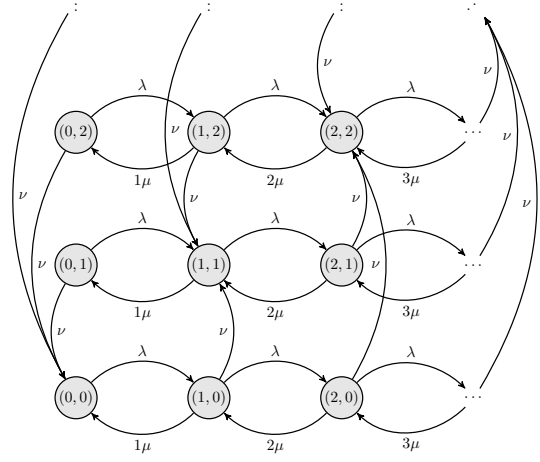
\begin{figure}[!htb]
  \centering
  \tikzstyle{input}=[font=\small\sffamily\bfseries]
\tikzstyle{rect}=[rectangle, draw=black, font=\small\sffamily\bfseries, inner sep=9pt]
\begin{tikzpicture}[->,>=stealth',auto,node distance=3cm,
  thick,main node/.style={fill=gray!20,draw, circle,inner sep=1pt}, ph node/.style={draw=none}]

  \foreach \x in {0,...,3}
    \foreach \y in {0,...,3} {
        \ifthenelse{\x=3 \OR \y=3}{
        		\ifthenelse{\x=3 \AND \y=3}{\node[draw=none] (\x\y) at (3*\x,3*\y) {\reflectbox{$\ddots$}};}{
		 	\ifthenelse{\x = 3}{\node[draw=none] (\x\y) at (3*\x,3*\y){$\cdots$};}{\node[draw=none] (\x\y) at (3*\x,3*\y){$\vdots$};};};
			}{\node [main node]  (\x\y) at (3*\x,3*\y) {$(\x, \y)$};};
       };

%

\foreach \x [count=\xi] in {0,...,2}
    \foreach \y [count=\yi] in {0,...,2}
      {\draw [->] (\x\y) to [out=45,in=135] node [pos=0.5,above]  {$\lambda$} (\xi\y);
        \draw [->] (\xi\y) to [out=225,in=315] node [pos=0.5, below] {$\xi \mu$}(\x\y);
       };
\foreach \x in {0,...,3}
        \foreach \y in {0,...,3}{
		 \ifthenelse{\x=\y}{}{\ifthenelse{\x > \y}{\draw [->] (\x\y) to [bend right] node [pos=0.5, left] {$\nu$}(\x\x);}{\draw [->] (\x\y) to [bend right] node [pos=0.5, right] {$\nu$}(\x\x);};};
	}

\end{tikzpicture}
    \caption{Rate diagram of the $M/M/\infty$ queue with Poisson observer.}
    \label{fig:transition_diag_2}
  \end{figure}

The balance equations for the stationary probabilities are
\begin{align}
	(\lambda + i\mu + \nu)\pi_{i,j} &= \lambda\pi_{i-1,j} + (i+1)\mu \pi_{i+1,j} + \1_{\{i=j\}}\nu\sum_{k}\pi_{j,k},  \,\, i \geq 1, \,\, j\geq 0; \nonumber \\
	(\lambda + \nu)\pi_{i,j} &= (i+1)\mu \pi_{i+1,j} + \1_{\{i=j\}}\nu\sum_{k}\pi_{j,k},  \,\, i = 0, \,\, j\geq 0, \nonumber
\end{align}

\begin{proposition}
Generating function $P(x,y)$ is solution of the functional equation
\begin{equation}
	(\nu + \lambda(1-x))P(x,y) + \mu (x-1) \px P(x,y) = \nu P(xy,1).
\label{eqn:gen_fun_2}
\end{equation}

The solution of functional equation (\ref{eqn:gen_fun_2}) is given by
\begin{equation}
P(x,y) = e^{\rho(x-1)} \cdot e^{\rho(y-1)} \cdot \sum_{k=0}^{\infty} \frac{\nu}{\nu + k \mu} \frac{\rho^k}{k!} (x-1)^k (y-1)^k.
\label{eqn:sol_gen_fun_2}
\end{equation}
\label{prop:sol_gen_fun_2}
\end{proposition}

\begin{proof}
Multiplying both sides of the rate balance equations by $x^iy^j$ and summing
over all possible $i$ and $j$ immediately leads to the equation
\[
	(\nu + \lambda)P(x,y) + \mu x \px P(x,y)
= \lambda x P(x,y) + \mu \px P(x,y)+ \nu P(xy,1),
\]
which can be alternatively written as (\ref{eqn:gen_fun_2}).

Now, use the fact that
\[
P(xy,1) = e^{\rho(xy-1)} = e^{\rho(x-1)(y-1)} \cdot e^{\rho(x-1)} \cdot e^{\rho(y-1)},\]
and introduce the function $h(x,y)$ via
\[
P(x,y) = h(x,y) \cdot  e^{\rho(x-1)} \cdot e^{\rho(y-1)}.
\]
Then, using (\ref{eqn:gen_fun_2}), we obtain
\begin{equation}
\nu h(x,y) + \mu (x-1) \px h(x,y) = \nu e^{\rho(x-1)(y-1)}.
\label{eqn:fun_h}
\end{equation}
If we now write
\[
h(x,y) = \sum_{k=0}^{\infty} \sum_{\ell=0}^{\infty} h_{k,\ell} (x-1)^k (y-1)^{\ell},
\]
then from (\ref{eqn:fun_h}) we obtain
\[
(\nu+\mu k) h_{k,\ell} = \left\{ \begin{array}{ll}
                                   \nu \frac{\rho^k}{k!} & \mbox{if $k=\ell$}, \\
                                   0             & \mbox{if $k \neq \ell$},
                                   \end{array}
                         \right.
\]
and hence
\[
h(x,y) = \sum_{k=0}^{\infty} \frac{\nu}{\nu + k \mu} \frac{\rho^k}{k!} (x-1)^k (y-1)^k.
\]
\end{proof}

\begin{corollary}[Asymptotics when $\nu \to \infty$ or $\nu \to 0$]
From (\ref{eqn:sol_gen_fun_2}) it follows that, if $\nu \to \infty$, then $P(x,y) \to e^{\rho(xy-1)}$ and, if $\nu \to 0$, then $P(x,y) \to e^{\rho(x-1)} \cdot e^{\rho(y-1)}$.
\end{corollary}

\begin{remark}
$P(x,y)$ is the probability generating function of the joint steady-state distribution of the queue length in an $M/M/\infty$ queue at two different time points, where the distance between the two time points is exponentially distributed with parameter $\nu$. Formula (\ref{eqn:sol_gen_fun_2}) can be alternatively derived using results on the transient distribution in the $M/M/\infty$ queue.
\end{remark}

\begin{remark}
An alternative expression for $P(x,y)$ is the following:
\begin{equation}
P(x,y)= e^{\rho(x-1)} \cdot e^{\rho(y-1)} \cdot
\int_{u=0}^1 \tfrac{\nu}{\mu} \, u^{\nu/\mu-1} \, e^{\rho u(y-1)(x-1)} \, du.
\label{eqn:alt_gen_fun_2}
\end{equation}
This can be derived from (\ref{eqn:gen_fun_2}) in the following way.
We have
\begin{eqnarray*}
\px \left( e^{-\rho x} (1-x)^{\nu/\mu} P(x,y) \right) &=& e^{-\rho x} (1-x)^{\nu/\mu} \left[ \px P(x,y) \right.\\
                &&\quad + \left.\frac{\nu + \lambda(1-x)}{\mu(x-1)} P(x,y)\right]\\
&=& e^{-\rho x} (1-x)^{\nu/\mu} \frac{\nu P(xy,1)}{\mu(x-1)} \\
&=& e^{-\rho x} (1-x)^{\nu/\mu} \frac{\nu}{\mu(x-1)} e^{\rho(xy-1)}.
\end{eqnarray*}
Hence,
\begin{equation}
e^{-\rho x} (1-x)^{\nu/\mu} P(x,y) = \int_{z=1}^x
\frac{\nu}{\mu(z-1)} (1-z)^{\nu/\mu}  e^{\rho(zy-1-z)} dz,
\end{equation}
which can alternatively be written as
\begin{equation}
 P(x,y) = e^{\rho (x-1)} \int_{z=1}^x
 \frac{\nu}{\mu(z-1)} \left(\frac{1-z}{1-x}\right)^{\nu/\mu} e^{\rho z(y-1)} dz.
\end{equation}
Now, substituting $u=\frac{z-1}{x-1}$ (and hence $z=(x-1)u+1$) yields (\ref{eqn:alt_gen_fun_2}).
\end{remark}

\section{Discussion}

We close the paper by briefly describing how our results help with decision-making in queueing systems. One application of this can be to model the trade-off between the monitoring costs and the suboptimality due to reduced monitoring. Control policies are sometimes obtained assuming changes can be made at arbitrary time instants (or equivalently assuming infinite monitoring frequency). In practice, however, monitoring the state incurs costs pushing one to monitor less frequently. Reducing the monitoring frequency lowers the measurement cost but also makes the policy more suboptimal. With the analysis in this paper, one can compute the performance obtained for a given monitoring frequency, $\nu$, and then determine the appropriate value of $\nu$ that optimizes the objective that accounts for both the performance as well as the monitoring costs.

 As an illustration, consider the problem of optimizing a linear combination of sojourn time and energy consumption that is mentioned in the introduction of the paper. The optimal policy is to set the speed to the number of jobs to the power of a coefficient (in this paper, we assume this coefficient to be $1$). In Fig. \ref{fig:eqs}, we plot the expected queue length as well as the expected speed as a function of $\nu$, the monitoring frequency, for $\bar{s} = 2$. It can be seen that both decrease with $\nu$ which points to a trade-off since a higher $\nu$ will entail a higher monitoring cost.

\begin{figure}[ht]
\label{fig:eqs}
\begin{center}
\centering
    \includegraphics[scale=0.7]{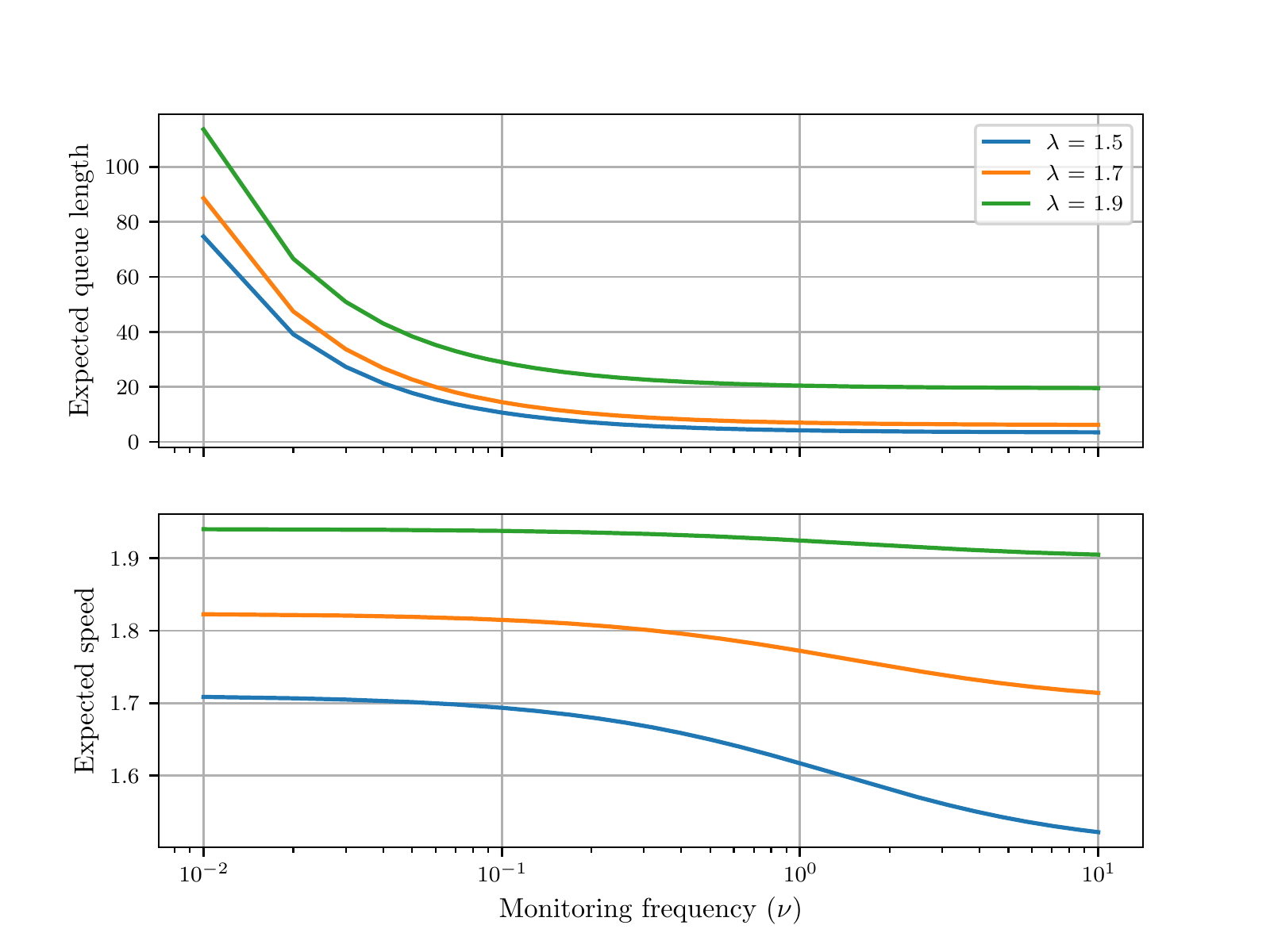}
\end{center}
\caption{Expected queue length and expected speed as functions of the monitoring frequency ($\nu$). $\bar{s} = 2$, $\mu = 1$.}
\end{figure}

We also remark that for the finite-speed case Prop \ref{prop:gf_fin} can easily be generalized to arbitrary increasing function speed profiles $s_i$, where $s_i$ is the value the speed is set to when state $i$ is observed.  For the above figures, it was assumed that $s_i = i\mu$.

\bibliographystyle{abbrvnat}
\bibliography{refs}
\end{document}


\tikzstyle{input}=[font=\small\sffamily\bfseries]
\tikzstyle{rect}=[rectangle, draw=black, font=\small\sffamily\bfseries, inner sep=9pt]
\begin{tikzpicture}[->,>=stealth',auto,node distance=3cm,
  thick,main node/.style={fill=gray!20,draw, circle,inner sep=1pt}, ph node/.style={draw=none}]
  
  \foreach \x in {0,...,5}
    \foreach \y in {0,...,2} {
        \ifthenelse{\x=5 \OR \y=3}{
		 	\ifthenelse{\x = 5}{\node[draw=none] (\x\y) at (3*\x,3*\y){$\cdots$};}{\node[draw=none] (\x\y) at (3*\x,3*\y){$\vdots$};};
			}{\node [main node]  (\x\y) at (3*\x,3*\y) {$(\x, \y)$};};
       };
      
%
  
\foreach \x [count=\xi] in {0,...,4}
    \foreach \y [count=\yi] in {0,...,2}  
      {\draw [->] (\x\y) to [out=45,in=135] node [pos=0.5,above]  {$\lambda$} (\xi\y);
           \draw [->] (\xi\y) to [out=225,in=315] node [pos=0.5, below] {$\y \mu$}(\x\y);
       };
\foreach \x in {0,...,4}
        \foreach \y in {0,...,2}{ 
		 \ifthenelse{\x=\y \OR \(\x > 2 \AND \y = 2\)}{}{\ifthenelse{\x > 2 \AND \y < 2}{\draw [->] (\x\y) to [bend right] node [pos=0.5, left] {$\nu$}(\x2);}{\draw [->] (\x\y) to [bend right] node [pos=0.5, right] {$\nu$}(\x\x);};};
	};           
\end{tikzpicture}